\documentclass{article}
\usepackage{graphicx,amssymb,amsmath}

\usepackage{color}
\usepackage{amsfonts}
\usepackage{latexsym}
\usepackage{fullpage}

\usepackage[usenames,dvipsnames]{xcolor}
\usepackage[colorlinks,citecolor=blue,linkcolor=BrickRed]{hyperref}

\usepackage[ruled]{algorithm}
\usepackage[noend]{algpseudocode}

\usepackage{comment}
\usepackage{graphicx,tipa}

\newcommand{\ignore}[1]{}

\usepackage{enumitem}

\usepackage[section]{placeins}

\usepackage{subfig}
\usepackage{setspace}

\newcommand{\sinn}[1]{\sin \left({#1}\right)}

\newcommand{\reals}{\mathbb{R}}

\algdef{SE}[DOWHILE]{Do}{doWhile}{\algorithmicdo}[1]{\algorithmicwhile\ #1}%

\def\rb{{\rho_\beta}}
\def\crb{{\lceil \rb \rceil}}

\newcommand{\alg}[1]{#1-Jump Algorithm}

\newtheorem{lemma}{Lemma}
\newtheorem{theorem}{Theorem}

\newenvironment{proof}{\noindent{\bf Proof.}}{\hfill \qed \vskip 5pt}
\def\qed{\hfill\rule{2mm}{2mm}}

\begin{document}

\title{Searching with Advice: Robot Fence-Jumping\footnote{This is the full version of the paper with the same title which will appear in the proceedings of the 28th CCCG (Canadian Conference on Computational Geometry), August 3-5, 2016, Vancouver.}
}

\author{Kostantinos Georgiou\thanks{Research supported in part by NSERC.}
\thanks{Department of Mathematics,
	Ryerson University,{\tt  konstantinos@ryerson.ca}
	}
\and
Evangelos Kranakis$^\dag$\thanks{School of Computer Science, 
	Carleton University,{\tt kranakis@scs.carleton.ca}}
\and
Alexandra Steau\thanks{School of Computer Science, 
	Carleton University,{\tt AlexandraSteau@cmail.carleton.ca}}
}

\index{Georgiou, Kostantinos}
\index{Kranakis, Evangelos}
\index{Steau, Alexandra}

\thispagestyle{empty}
\maketitle

\begin{abstract}
We study a new search problem on the plane involving a robot and an immobile treasure, initially placed at distance $1$ from each other. The length $\beta$ of an arc (a fence) within the perimeter of the corresponding circle, as well as the promise that the treasure is outside the fence, is given as part of the input. The goal is to device movement trajectories so that the robot locates the treasure in minimum time. Notably, although the presence of the fence limits searching uncertainty, the location of the fence is unknown, and in the worst case analysis is determined adversarially. Nevertheless, the robot has the ability to move in the interior of the circle. In particular the robot can attempt a number of chord-jump moves if it happens to be within the fence or if an endpoint of the fence is discovered. 

The optimal solution to our question can be obtained as a solution to a complicated optimization problem, which involves trigonometric functions, and trigonometric equations that do not admit closed form solutions. For the \alg{1}, we fully describe the optimal trajectory, and provide an analysis of the associated cost as a function of $\beta$. Our analysis indicates that the optimal \alg{k} requires that the robot has enough memory and computation power to compute the optimal chord-jumps. Motivated by this, we give an abstract performance analysis for every \alg{k}. Subsequently, we present a highly efficient Halving Heuristic \alg{k} that can effectively approximate the optimal \alg{k}, with very limited memory and computation requirements.

\vspace{0.5cm}
\noindent
{\bf Key words and phrases.}
Disk,
Fence,
Optimization,
Robot,
Search,
Speed,
Treasure.
\end{abstract}


\section{Introduction}

Geometric search is concerned with finding a target placed in a geometric region and has been investigated in many areas of mathematics, theoretical computer science, and robotics.  In each instance one aims to provide search algorithms that optimize a certain cost, which may take into account a variety of important characteristics and features of the domain, computational abilities of the searcher, assumptions about the target, etc. In this paper, we introduce and study {\em fence-jumping search}, a new search problem involving a robot, an unknown stationary fence (barrier), and an unknown stationary target (or treasure) in the plane. 

The location of the treasure is unknown to the robot. However, it has knowledge that at the start it is located at a distance of $1$ (unit) away from the treasure. Equivalently, the treasure is stationed on the perimeter of a disk (within the known environment), which is centered at the start point of the robot. A {\em fence}, a given circular arc of  length $\beta$, is placed on the perimeter of the disk, whose location is also unknown to the robot. Further, the robot has the knowledge that the treasure is located on the perimeter but not on the fence.  Depending on its trajectory, the robot may move along the perimeter of the circle and occasionally, say when within the fence, it may
want to move along a chord, or as we say to {\em fence-jump}, so as to reduce the time necessary to perform the search. We will analyze several fence-jumping algorithms that will allow us to reach the treasure in minimal time.

We study the fence-jumping search problem for one robot starting at the center of the disk and moving at a constant speed $1$. We assume the treasure is stationary and placed by an adversary at the beginning of each round depending on the fence location. The adversary positions the treasure on the perimeter, but outside the fence. The robot may move anywhere on the disk in an attempt to find this treasure; it is also able to use geometric knowledge so as to decrease the amount of time spent during the search. That is to say, since the robot knows that the treasure is not located on the fence, it could try to bypass it by ``jumping over the fence''. Goal of this paper is to determine a trajectory so that the robot finds the treasure in optimal time.

\subsection{Related Work}


The type of search problem investigated in our work was first seen sixty years ago when Beck~\cite{beck1964linear}  and Bellman~\cite{bellman1963optimal} asked an important, yet simplistic question tied to the minimization of distance.
Motivated from this, several different natural search problems
have been studied including the use of a fixed~\cite{czyzowicz2014evacuating,koopman1957theory} or mobile target~\cite{stone1974search}, the tools searchers have access to, the number of searchers, the communication restrictions and many more. Often, the essential part of the robot activity is the recognition and/or mapping of the terrain. In the case of a known structure, the main objective of the search is to minimize the time to find the treasure. Searching for a motionless target has been studied in the cow-path problem~\cite{beck1964linear}, lost in a forest problem~\cite{finch2004lost,isbell1957optimal} and plane searching problem~\cite{baeza1995parallel,baezayates1993searching}.  

Baeza-Yates \emph{et al} in their well known paper~\cite{baezayates1993searching} study the worst-case time for search involving one robot and a treasure at an unknown location in the plane, such as a simple line. 
Useful surveys on search theory can also be found 
in~\cite{benkoski1991survey}~and~\cite{dobbie1968survey}.

Search by multiple robots with communication capabilities has been considered in~\cite{dobrev2001mobile,hoffmann2001polygon}, while \cite{czyzowicz2014evacuating,czyzowicz2015evacuating} study the evacuation of $k$ robots searching for an exit  located on the perimeter of a disk. 
The problem of finding trajectories for obstacle avoidance in both known and unknown terrains has been considered in several papers including  \cite{badal1994practical,blum1991navigating,pozna2009design}.


\subsection{Outline and Results of the Paper}

As a main objective, our approach will have to design algorithms for finding the treasure in optimal time, while adapting to the fence structure located on the perimeter. Thus, leading us to propose algorithms that attempt to deliver the optimal shortcuts necessary to exit and/or avoid the fence structure. 

An outline of our results is as follows. In Section~\ref{sec:0-jump}, we introduce the basic concepts and analyze a simple search algorithm for finding the treasure without involving any jumps. In Section~\ref{sec:1-jump}, we introduce and analyze the optimal $1$-Jump search algorithm. Meanwhile, in Section~\ref{sec:k-jump}, we propose a generic description of $k$-Jump algorithms. In Section~\ref{sec: halving algo}, we study a $k$-Jump algorithm based on a {\em halving} heuristic, which approximates the optimal jump without relying on solutions of trigonometric optimization problems. In Section~\ref{sec: comparison}, we contrast the choices and performance of the Halving \alg{k} with the choices and performance of the Optimal \alg{k} that was obtained using optimization software packages, for $k\leq 3$. 
We conclude with Section~\ref{sec:conclusion}. 



\section{Preliminary Observations}\label{sec: Prel}
\label{sec:0-jump}

First we introduce the basic concepts and assumptions of our model. Initially, we make the assumption that the robot is located in the center of the disk with a radius of $1$, and a treasure is located at distance of $1$ from the robot, on the perimeter of the disk. We define this treasure to be a point on the disk and, thus, does not take any space on the perimeter. The treasure location is always unknown to the robot until it moves directly over its point location. That is to say, the robot has no vision capabilities, in that it becomes aware of what each point on the circle is, i.e. a fence point, treasure or nothing special, only if the point is visited. The robot moves at the same speed throughout its search on the disk and the movement of the robot from the center always takes $1$ unit of time. The robot has the computational power to numerically solve trigonometric equations through the use of deterministic processors which possess the required memory for these processes. 
 
Recall that goal of the robot is to optimize the length of its trajectory using various types of movements, i.e. the robot may walk on the fence or even jump over the fence moving along a chord (within the interior of the circle).

To begin we provide a naive solution to our treasure finding problem, which we will then improve with a number of algorithms. In what follows, we denote the length of the fence by $\beta$, given as part of the input. Independently of the algorithm considered, any deterministic algorithm will first have the robot move to an arbitrary point on the perimeter of the circle, thereafter referred to as the \textit{basic landing point}, with the intention that the robot will start moving/searching the circle counterclockwise, which is further examined in Algorithm~\ref{alg:0jumpa}. 

\begin{algorithm}[H]
\caption{0-Jump Algorithm}\label{alg:0jumpa}
\begin{algorithmic}[1]
\State  {Walk to the perimeter of the disk}
\State {Continue walking on the perimeter counterclockwise}
\If {you reach fence endpoint}
\State {Jump along the corresponding chord of length}
\State {$2\sinn{\beta/2}$}
\Else
\State {Walk on perimeter until you find treasure}
\EndIf
\end{algorithmic}
\end{algorithm}

Our work provides a focus on algorithms that perform well under worst case analysis. As such, the performance of any algorithm will be determined after an adversary decides on the location of both the basic landing point, the fence itself, and the treasure. For the sake of exposition, we now present the worst case termination time depending on the location of the basic landing point.

\begin{lemma}\label{lem: 0-jump performance}
The worst case termination time of Algorithm~\ref{alg:0jumpa} is\footnote{The usefulness of notation $c^i_j$ for the cost of the algorithm will be transparent in the next sections}
$$
\left\{
\begin{array}{ll}
c^0_0~:=~1+2\pi -\beta+2\sinn{\beta/2}&~\mbox{Figure~\ref{fig:land1}}  \\
c^0_1~:=~1+2\pi&~\mbox{Figure~\ref{fig:land2}}
\end{array}
\right.
$$
where the reference in the right column indicates the Figure which applies to the case.

\begin{figure}[!htb]
\centering
\subfloat[Landing outside fence.]{\label{fig:land1}{\includegraphics[height=6cm]{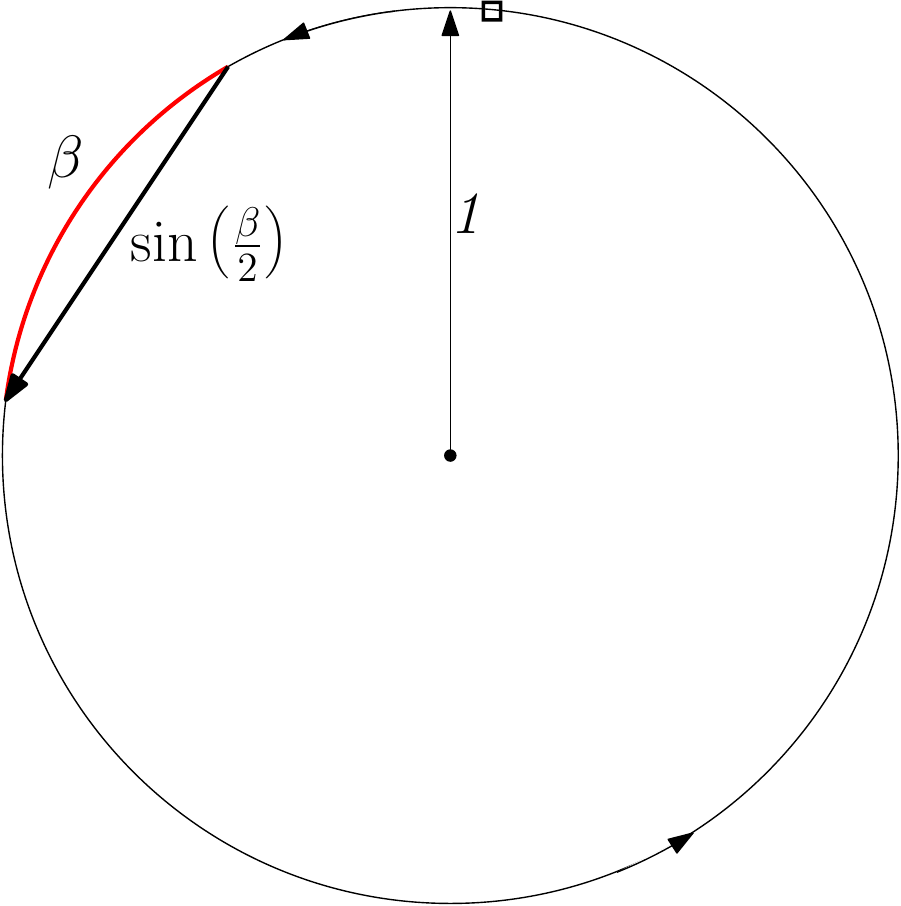}}}
\hspace{1em}
\subfloat[Landing within fence.]{\label{fig:land2}{\includegraphics[height=6cm]{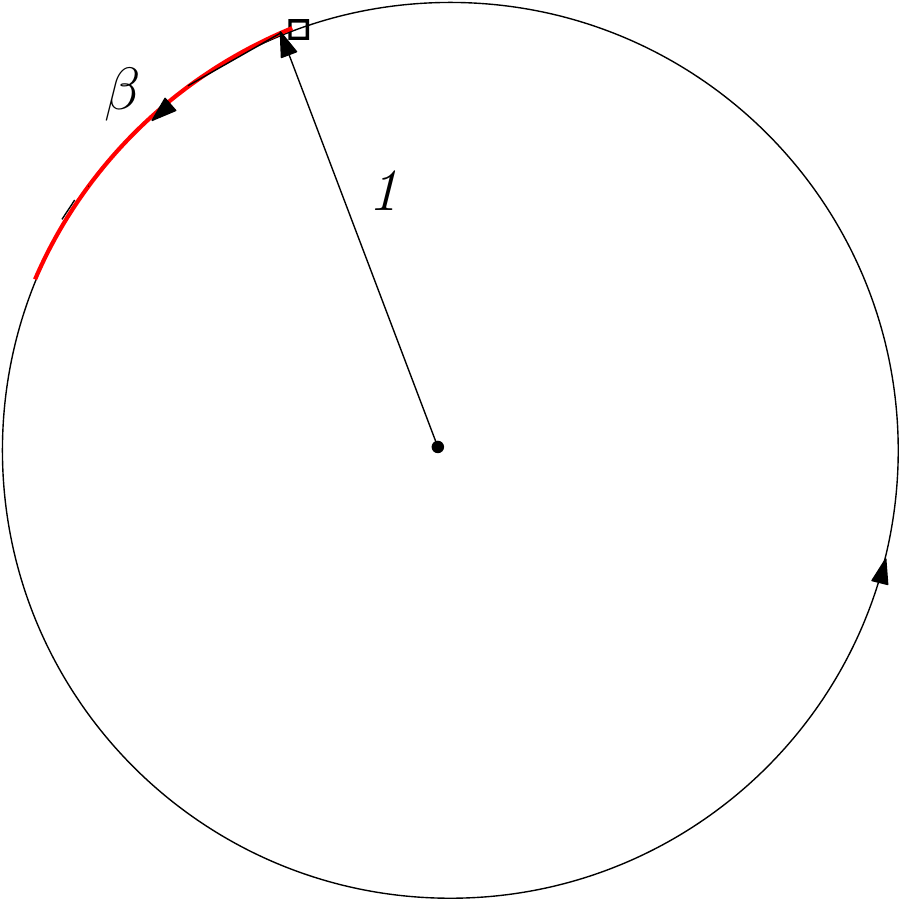}}}
\caption{Basic landing point.}
\label{fig:caseshave}
\end{figure}
\end{lemma}
\begin{proof} 
Suppose that the basic landing point is outside the fence, as seen in Figure~\ref{fig:land1}, and say that the clockwise distance between the landing point and the fence is $x\in (0, 2\pi-\beta)$. It is straightforward that the adversary would place the treasure clockwise inbetween the landing point and the fence, at clockwise distance $y\in (0,x)$ from the landing point. Then, for all $x \in (0, 2\pi-\beta)$ the cost of the algorithm would be  
$$
\sup_{y\in (0,x)}\{1+2\pi -\beta+2\sinn{\beta/2}-y\}=1+2\pi -\beta+2\sinn{\beta/2}.
$$

In the other case, the landing point is within the fence, as illustrated in Figure~\ref{fig:land2}. Suppose that the clockwise distance between the landing point and the endpoint of the fence is $x\in (0,\beta)$. Also suppose that the clockwise distance between the same fence endpoint and the treasure is $y\in (0,2\pi-\beta)$. Then, the robot will locate the treasure in time 
$$\sup_{x,y}\{1+2\pi -x-y\}=1+2\pi.$$
This proves Lemma~\ref{lem: 0-jump performance}.
\end{proof}

It is intuitive that having the basic landing point outside the fence is a ``favorable event'' in that for all $\beta$, $c^0_0\leq c^0_1$. This follows formally from the fact that the non-negative expression $\beta-2\sinn{\beta/2}$ is increasing in $\beta>0$. Hence, the performance of Algorithm~\ref{alg:0jumpa} is $1+2\pi$.

Next, we focus on algorithms that can address the choice of the adversary placing (basic) landing points within the fence. In such algorithms the robot will try to jump in an attempt to land outside the fence. 

\FloatBarrier

\section{The Optimal 1-Jump Algorithm}\label{sec: opt 1-jump}
\label{sec:1-jump}

In this section we analyze the optimal \alg{1}, which also serves as a warm-up for the analysis of the generic k-Jump Algorithm. \alg{1}s are fully determined by the (unique) chord jump of corresponding arc-length $\alpha$ they make in case the basic landing point (of the robot) is within the fence. 

It is worthwhile discussing the required specifications for the algorithm to be correct. First, we require the jump to be in ``counter-clockwise'' direction,
i.e. that $\alpha \leq \pi$ (this also breaks the symmetry for the adversarial placements of the fence and the treasure). Second, we further require that the chord jump does not pass over the area that could hold the treasure, landing back to the fence. For that, it is of importance that $\alpha \leq 2\pi -\beta$. To summarize, the 1-Jump Algorithm is fully determined by choosing $\alpha$ satisfying 
$$
0< \alpha \leq \min \{ \pi, 2\pi - \beta\}.
$$
To resume, Algorithm~\ref{alg:o1Jump} with parameter $\alpha$ runs similarly to Algorithm~\ref{alg:0jumpa}, except from the case that the last landing of Algorithm~\ref{alg:0jumpa} (which happens to be the basic one) is within the fence. If that happens, Algorithm~\ref{alg:o1Jump} makes a counterclockwise jump corresponding to arc length $\alpha$. If the 1st-jump landing point is in the fence, then it runs Algorithm~\ref{alg:0jumpa}. Otherwise, the 1st-jump landing point is outside the fence and the robot applies the following \textit{remedy phase}: move clockwise along the periphery of the circle till the endpoint of the fence is found, say at arc distance $x$, and then return to the 1st-jump landing point along the corresponding chord of length $2\sinn{x/2}$, and continue executing Algorithm~\ref{alg:0jumpa}. 

\begin{figure}[!htb]
\centering
\subfloat[1st jump landing is outside fence.]{\label{fig:optalp1}{\includegraphics[height=6cm]{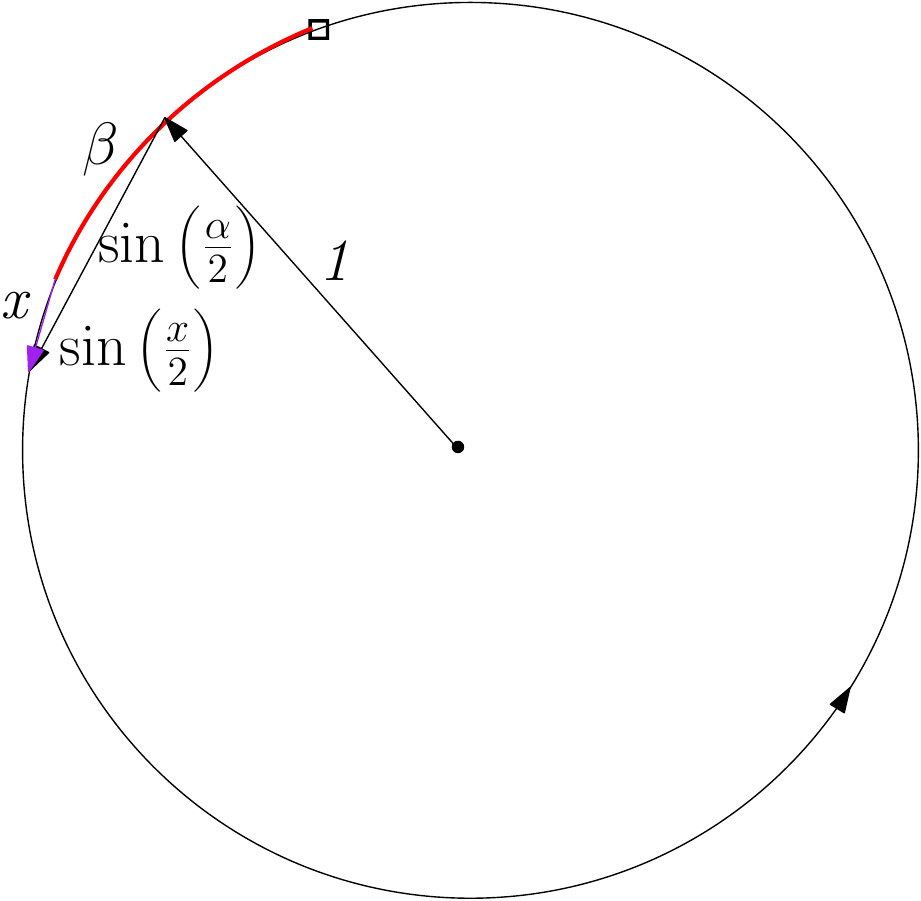}}}
\hspace{1em}
\subfloat[1st jump landing is inside fence.]{\label{fig:optalp2}{\includegraphics[height=6cm]{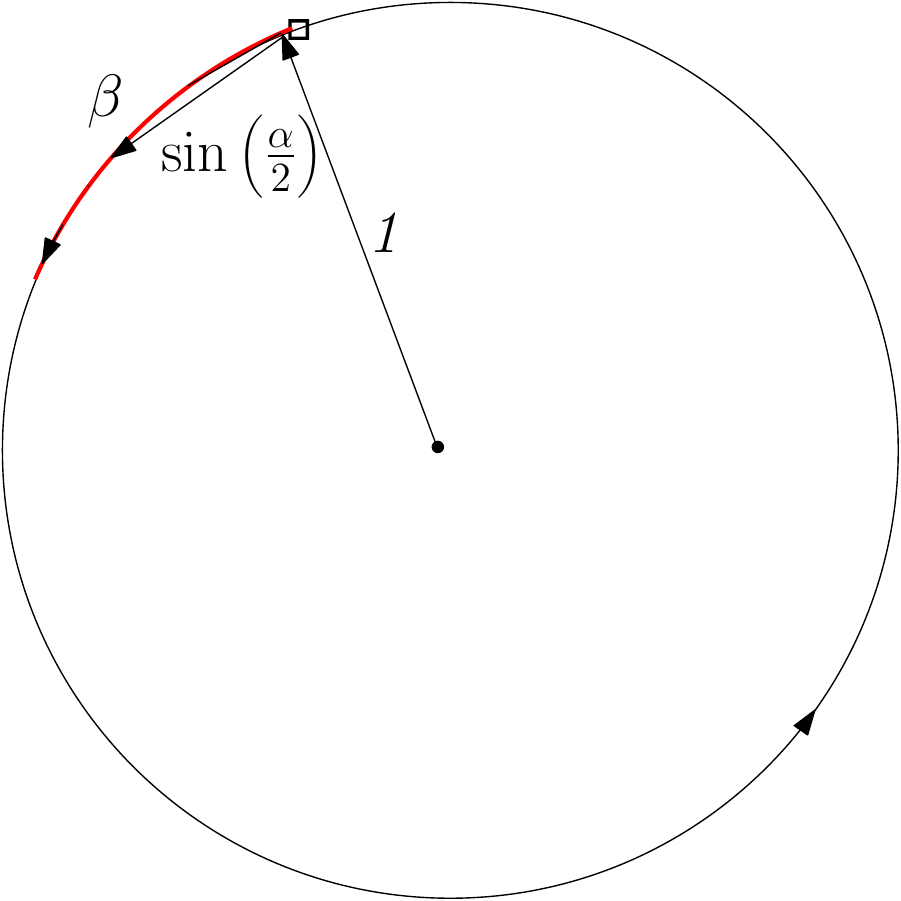}}}
\caption{\alg{1} basic landing point.}
\label{fig:caseshave1}
\end{figure}

\begin{algorithm} [H]
\caption{1-Jump}\label{alg:o1Jump}
\begin{algorithmic}[1]
\State  {Walk to the perimeter of the disk}
\If {your landing point is inside the fence}
\State {make a ccw chord jump of arc length $\alpha$}
\EndIf
\State {Perform Algorithm~\ref{alg:0jumpa}}
\end{algorithmic}
\end{algorithm} 

\begin{lemma}\label{lem: cost of 1-jump alg}
Depending on the landing points, the cost of Algorithm~\ref{alg:o1Jump} with parameter $\alpha$ is 
\begin{equation}
\left\{
\begin{array}{ll}
c^1_0~:=~ 1+2\pi - \beta + 2\sinn{\beta/2}& \mbox{Figure~\ref{fig:land1}}  \\
c^1_1~:=~ c^1_0 + 4 \sinn{\alpha/2} - 2\sinn{\beta/2} & \mbox{Figure~\ref{fig:optalp1}}  \\
c^1_2~:=~ 1+2\pi - \left(\alpha- 2\sinn{\alpha/2} \right) & \mbox{Figure~\ref{fig:optalp2}} 
\end{array}
\right.
\end{equation}
with the understanding that $c^1_0, c^1_1, c^1_2$ are functions on $\beta$ and $\alpha$,
and the reference in the right column indicates the Figure which applies to the case.
\end{lemma}
\begin{proof} 
Clearly, if the basic landing point is in the fence, then the cost of Algorithm~\ref{alg:o1Jump} $c^1_0$ is equal to cost $c^0_0$ of the Algorithm~\ref{alg:0jumpa} (for the same case). 

Suppose now that the basic landing point is in the fence. Algorithm~\ref{alg:o1Jump} performs a counterclockwise chord jump of length $2\sinn{\alpha/2}$. We examine two more subcases. In the first subcase, the 1st-jump landing point is outside the fence as seen in Figure~\ref{fig:optalp1}, say at clockwise distance $x\in (0,\alpha)$ from the fence. Then the robot follows the remedy phase spending $x+2\sinn{x/2}$ more time to come back to the same landing point. Clearly, the worst positioning of the treasure is to be arbitrarily clockwise close to the fence. That would make the robot search for an additional time of $2\pi - \beta - x$ for a total of 
$$1+x+2\sinn{x/2}+2\pi - \beta - x = 1+2\sinn{x/2}+2\pi - \beta.$$
Since $x\leq \alpha \leq \min \{ \pi, 2\pi - \beta\}$ and by the monotonicity of $\sinn{x/2}$ we see, as promised, that the cost in that case is no more than
\begin{align*}
& \sup_{0<x< \alpha} \{1+2\sinn{\alpha/2}+2\sinn{x/2}+2\pi - \beta\} \\
&~~~~~~~~~~~~~~~= 1+4\sinn{\alpha/2}+2\pi - \beta \\
&~~~~~~~~~~~~~~~= c^1_0 + 4 \sinn{\alpha/2} - 2\sinn{\beta/2} 
\end{align*}

In the second subcase the 1st-jump landing point is in the fence and is illustrated in Figure~\ref{fig:optalp2}. It is not difficult to see that the worst configuration in this case is when the robot's basic landing point is arbitrarily close to the clockwise endpoint of the fence, while the treasure is arbitrarily close to the same endpoint and outside the fence. Clearly, the robot in that case traverses the whole circle, saving only an arc of length $\alpha$ which is jumped over using the corresponding chord of length $2\sinn{\alpha/2}$. Overall, the cost in this case becomes $1+2\pi - \alpha + \sinn{\alpha/2} $. This completes the proof of Lemma~\ref{lem: cost of 1-jump alg}.
\end{proof}

Critical to our analysis toward specifying the optimal choice of $\alpha$ is the solution to a specific equation that does not admit a closed form.  Consider expression $\alpha + 2\sinn{\alpha/2}$ which is monotonically increasing. As such, for every $\beta \in \reals$, the equation $\alpha + 2\sinn{\alpha/2} = \beta$ admits a unique solution in $\alpha$. Motivated by this observation we write that 
``$\alpha_\beta$ is the unique real number satisfying equation $\alpha_\beta + 2\sinn{\alpha_\beta/2} = \beta$''.
\ignore{ 
test[b_] := a /. FindRoot[ a + 2*Sin[a/2] == b, {a, 1}]
}

Moreover, since $\alpha + 2\sinn{\alpha/2}$ is increasing in the variable $\alpha$, so is $\alpha_\beta$ in the variable $\beta$. 
We are now ready to define and analyze the optimal 1-Jump Algorithm.

\begin{theorem}\label{thm: opt 1-jump}
Let $\gamma$ be the unique solution to equation $\pi=\gamma - \sinn{\gamma/2}$ ($\gamma \approx 4.04196$).
\ignore{
In[4]:= FindRoot[ Pi == x - Sin[x/2], {x, 4}]
Out[4]= {x -> 4.04196}
}
The optimal 1-Jump Algorithm chooses jump step corresponding to arc length $\alpha = \alpha_\beta$ if $\beta \leq \gamma$, $\alpha = 2\pi-\beta$ if $\beta > \gamma$
\ignore{$$
\alpha = 
\left\{
\begin{array}{ll}
\alpha_\beta& ~\mbox{if}~\beta \leq \gamma \\
2\pi-\beta& ~\mbox{if}~\beta > \gamma
\end{array}
\right.
$$}
and terminates in time 
\begin{equation}
1+ 
\left\{
\begin{array}{ll}
2\pi - \alpha_\beta + 2\sinn{\alpha_\beta/2}& ~\mbox{if}~\beta \leq \gamma \\
\beta+2\sinn{\beta/2}& ~\mbox{if}~\beta > \gamma.
\end{array}
\right.
\end{equation}
\end{theorem}
\begin{proof} 
By Lemma~\ref{lem: cost of 1-jump alg}, the optimal 1-Jump Algorithm is determined by choosing $\alpha$ that minimizes 
$$\sup_{0 < \alpha < \min \{ \pi, \pi - \beta\}} \{ c^1_0, c^1_1(\alpha), c^1_2(\alpha)\},$$
where in the expression above, we make the dependence on $\alpha$ explicit. 
Again, it should be clear that having the basic landing point outside the fence is a ``favorable event''. Intuitively, this is the only case that the robot makes full use of the fact that the treasure does not lie within the fence, jumping over it and using the corresponding chord. Effectively, this implies that for all $\beta, \alpha$ we have  $c^1_0\leq \min\{ c^1_1(\alpha), c^1_2(\alpha)\}$.

\ignore{ 
test[b_] := a /. FindRoot[ a + 2*Sin[a/2] == b, {a, 1}]
}

Next, for any $\beta \in (0, 2\pi)$ we need to choose $\alpha$ so as to minimize $\max\{c^1_1(\alpha), c^1_2(\alpha)\}$. To that end, note that  $c^1_1, c^1_2$ exhibit different monotonicities with respect to $\alpha$ so that, if possible, the minimum will be attained when the two costs are equal. 
Equating the two costs gives that $\alpha + 2\sinn{\alpha/2}=\beta$. Recall that we have denoted the unique solution to the equation by $\alpha_\beta$ which is increasing in $\beta$. Since the jump step needs to stay no more than $\min\{\pi,2\pi-\beta\}$, the choice $\alpha=\alpha_\beta$ (which is the best possible) is valid only when $\alpha_\beta \leq \min \{\pi, 2\pi-\beta\}$.  Numerically we can compute $\alpha_\pi \approx 1.66$, which due to the monotonicity of $\alpha_\beta$ implies that the dominant constraint is that $\alpha_\beta \leq 2\pi-\beta$, and hence any restrictions will be imposed for $\beta >\pi$. Indeed, setting $\alpha_\beta=2\pi-\beta$, and substituting in $\alpha_\beta + 2\sinn{\alpha_\beta/2}=\beta$ we obtain $2\pi-\beta + 2\sinn{\pi-\beta/2}=\beta$. The value of $\beta$ that satisfies this equation is $\gamma \approx 4.04196$.  

To resume, as long as $\beta \leq \gamma$, the best choice for the jump is the solution to the equation $\alpha + 2\sinn{\alpha/2}=\beta$. When $\beta > \gamma$, the best jump step is equal to $2\pi-\beta$. 

From the discussion above, the induced cost when $\beta \leq \gamma$ would be equal to $c^1_1(\alpha_\beta)$, as it reads in Lemma~\ref{lem: cost of 1-jump alg}. Finally, when $\beta>\gamma$ the induced cost would be
\begin{align*}
&\max \{ c^1_1(2\pi-\beta), c^1_2(2\pi-\beta)\} \\
&~~~= 1+ 2\pi  + \max \{ 4\sinn{\pi-\beta/2}-\beta,  2\sinn{\pi-\beta/2} - 2\pi +\beta \} \\
&~~~= 1+ 2\pi  +2\sinn{\beta/2} +  \max \{ 2\sinn{\beta/2}-\beta,  - 2\pi +\beta \} \\
&~~~= 1+ \beta +2\sinn{\beta/2}
\end{align*}
where the last equality is due to that $\beta \geq \gamma$, the definition of $\gamma$ and the fact that $-2\pi+\beta$ is increasing in $\beta$. This proves Theorem~\ref{thm: opt 1-jump}.
\end{proof}

Notably, the proof of Theorem~\ref{thm: opt 1-jump} suggests that for the best strategy $\alpha$ as a function of $\beta$, we have that $c^1_0\leq c^1_2(\alpha) \leq c^1_1(\alpha)$. This was expected, since having the basic landing outside the fence is intuitively more favourable than having it inside the fence and without needing the remedy phase, which is more favourable than needing the remedy phase. It is also interesting to note that for $\beta \leq \gamma$, the best jump choice $\alpha_\beta$ attains values close to $\beta/2$. This suggests an alternative approach to the problem that does not require the ability to solve technical trigonometric equations, and that will be explored later. Finally, there is a nice suggested recursive relation between costs $c^1_0, c^1_1$ that is soon to be generalized for k-Jump Algorithms.

\section{Generic Description of \alg{$k$}s}\label{sec: generic k-alg}
\label{sec:k-jump}

Analogously to the previous sections, the \alg{k} has parameters $\alpha_1, \ldots, \alpha_k$ and runs similarly to the \alg{(k-1)}, except from the case that the last landing point of the \alg{(k-1)} (which happens to be the (k-1)st-jump landing point, if this is realized) is within the fence. If that happens, the \alg{k} makes an additional counterclockwise jump corresponding to arc length $\alpha_k$. If the $k$th-jump landing point is in the fence, then it runs Algorithm~\ref{alg:0jumpa}. Otherwise, the $k$th-jump landing point is outside the fence,  and the robot applies the remedy phase from Algorithm~\ref{alg:o1Jump} in Section~\ref{sec: opt 1-jump}. 

\begin{algorithm}[H]
\caption{\alg{k}}\label{alg:kjumpA}
\begin{algorithmic}[1]
\State  {Walk to the perimeter of the disk}
\State { $i\leftarrow 0$}
\While{landing point is inside the fence \& $i<k$}
\State { $i\leftarrow i+1$}
\State {make a ccw chord jump of arc length $\alpha_i$}
\EndWhile
\State {Perform Algorithm~\ref{alg:0jumpa}}
\end{algorithmic}
\end{algorithm}

It is clear from the discussion above that any \alg{k} is specified by the jump steps $\alpha_1, \alpha_2, \ldots, \alpha_k$, where the $i$th jump is realized only if the basic landing point, along with the landing points of the previous $i-1$ jumps fall within the fence. In order to preclude the possibility that a jump passes over the area that holds the treasure and bring the robot back to the fence we require that $\alpha_i \leq \beta$. Moreover, for the jumps to be in counterclockwise direction (and to break the symmetry) we also require that $\alpha_i \leq \pi$.

Similarly, for the \alg{1} we required that $\alpha_1 \leq \min \{\pi, 2\pi-\beta\}$. However, according to Theorem~\ref{thm: opt 1-jump}, the optimal jump step is less than $\beta$ (for all $\beta$), meaning that the correctness condition for choosing the jump step could have been replaced by $\alpha_1 \leq \min \{\beta, 2\pi-\beta\}$. Indeed, our intuition tells us that an algorithm, which after the basic landing point within the fence makes a jump more than the length of the fence, will land outside the fence and subsequently will need unavoidably to apply the (suboptimal) remedy phase. Motivated by this observation, we require the following condition regarding the step sizes of \alg{k}'s:
$
\alpha_i \leq \min \left\{\pi, 2\pi-\beta \right\}, ~~i=1, \ldots, k.
$

The next lemma generalizes Lemma~\ref{lem: cost of 1-jump alg} and provides a handy recurrence description of the cost of the \alg{k} with jump steps $\alpha_1, \ldots, \alpha_k$ depending on the first landing point outside the fence. In this direction, we denote by $c^k_t$ to be the worst case cost of the \alg{k} when the basic landing point along with the landing points of the first $t-1$ jumps fall all inside the fence and the robot lands outside the fence in the the $t$th jump, which is shown in Figure~\ref{fig:multiple}.
Let us observe that, $c^k_0$ is the cost of the case when the basic landing point is outside the fence, while $c^k_{k+1}$ corresponds to the case that the landing points of all $k$ jumps, as well as the basic landing point, fall inside the fence. \\ 

\begin{figure}[!htb]
	\centering
	\includegraphics[height=7cm]{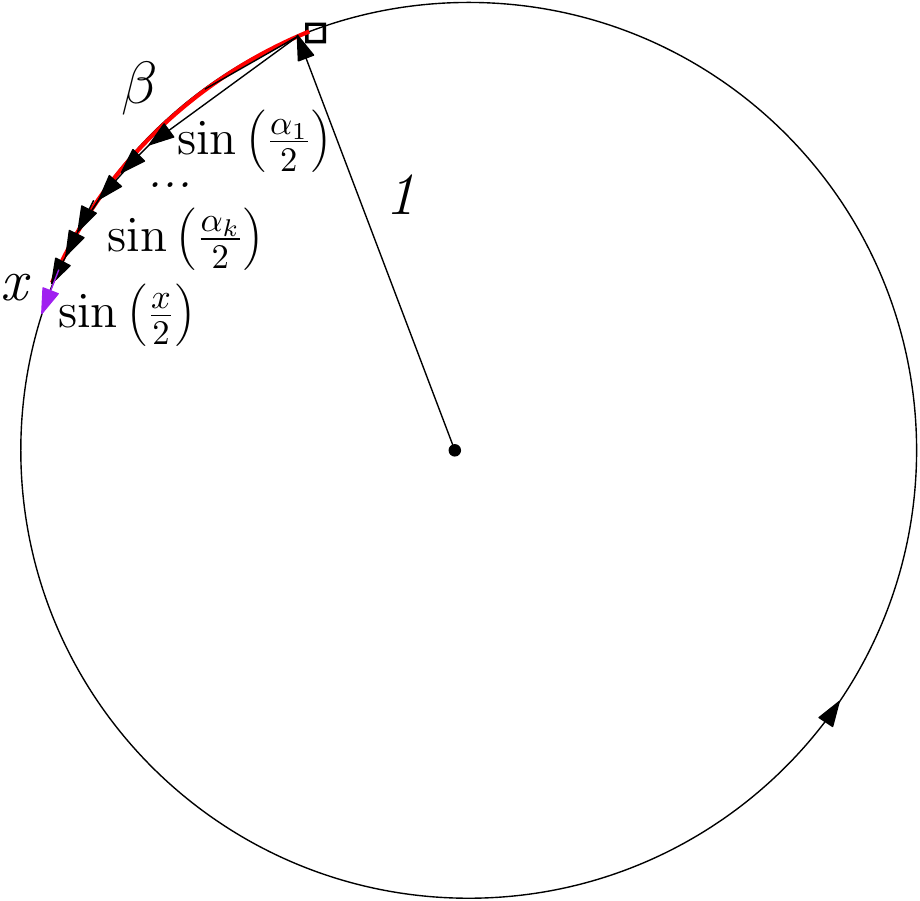}
	\caption{\alg{k}}
	\label{fig:multiple}
\end{figure}

\begin{lemma}\label{lem: cost of k-jump alg}
For any $\beta$, let $\alpha_0=\beta$. Depending on the landing points, the cost of the \alg{k} with jump steps $\alpha_1, \ldots, \alpha_k$ is 
\begin{equation}
\left\{
\begin{array}{ll}
c^k_0&:=~ 1+2\pi - \alpha_0 + 2\sinn{\alpha_0/2}\\
c^k_t&:=~ c^k_{t-1} + 4 \sinn{\alpha_t/2} - 2\sinn{\alpha_{t-1}/2} \\
c^k_{k+1}&:=~ 1+2\pi - \sum_{i=1}^k \left(\alpha_i- 2\sinn{\alpha_i /2} \right) \\ 
\end{array}
\right.
\end{equation}
with the understanding that $c^k_t$ are functions on $\beta$ and $\alpha_1, \ldots, \alpha_{t-1}$, for $t=1, \ldots k+1$. 
\end{lemma}
\begin{proof} 
As previously mentioned, when the basic landing point is outside the fence, the cost is indeed $c^k_0=c^0_0$. Furthermore, when all landing points, including the basic one, fall within the fence, then similarly to the cost $c^1_2$ of Lemma~\ref{lem: cost of 1-jump alg}, the worst positioning of the fence makes the basic landing point inside and arbitrarily close to the counterclockwise endpoint of the fence. Meanwhile, the treasure is arbitrarily close to the same endpoint but outside the fence. Effectively, the robot in this case will traverse the entire circle counterclockwise, saving from each jump exactly $\alpha_i-2\sinn{\alpha_i/2}$, $i=1, \ldots, k$. 

For the most interesting case, we need to compare the costs $c^k_t, c^k_{t-1}$, for some $t \in \{1, \ldots, k\}$. In both cases, the worst positioning of the treasure is arbitrarily close to the clockwise endpoint of the fence. The worst positioning of the basic landing point should bring the robot inside, as well as arbitrarily close to the counterclockwise endpoint of the fence, so as to induce the maximum possible remedy phase cost. Note that for the case of cost $c^k_t$, the robot traverses twice the chord of length $2\sinn{\alpha_t/2}$, but only once the chord of length $2\sinn{\alpha_{t-1}/2}$. Other than that, in both cases, the robot perform exactly the same jumps, and search exactly the same subperimeter of the circle.  This proves Lemma~\ref{lem: cost of k-jump alg}.
\end{proof}

\section{The Halving Heuristic \alg{$k$}}
\label{sec: halving algo}

In this section, we present a simple heuristic that is meant to approximate the optimal jump steps without relying on solutions of trigonometric optimization problems. Most importantly, our algorithm requires very limited memory and does not need to perform numerical operations other than simple algebraic manipulations. In fact, there are only constant many operations needed to determine every possible jump size. Moreover,  parameter $k$, i.e. the number of jumps, may not necessarily be determined in advance, and is allowed to be even infinite. First, we present the algorithm and analyze it. Then, in Section~\ref{sec: comparison}, we contrast it to the Optimal \alg{k} (for certain values of $k$). 

Closely examining the optimal solution for the \alg{1} in Section~\ref{sec: opt 1-jump}, we are tempted to choose an alternative first jump step equal to $\beta/2$, which is a good approximation to $\alpha_\beta$. This choice is valid, as long as the jump does not exceed $2\pi-\beta$, and indeed for large enough values of $\beta$, i.e. for $\beta \geq 4.041$, as per Theorem~\ref{thm: opt 1-jump}, the best choice for just one jump is $2\pi -\beta$. Note that changing the first jump from $\alpha_\beta$ to $\beta/2$ results to a new threshold value $\frac43\pi\approx 4.188$ after which the first jump should become $2\pi-\beta$. Interestingly, the pattern repeats also in the optimal \alg{k}s (see Section~\ref{sec: comparison}).

The previous observation suggests a natural heuristic for \alg{k}s. First, go to an arbitrary point on the circle. While in (some unknown position in) the fence, make a valid jump (i.e. no more than $\pi, 2\pi-\beta$) equal to half of the unexplored fence, unless this value exceeds $2\pi-\beta$ in which case the jump should be $2\pi-\beta$. Formally, the description of the heuristic follows if we can determine the length of the chord-jump $\alpha_i$ in every $i$-th jump, and then invoke Algorithm~\ref{alg:kjumpA}.

\begin{algorithm}[H]
\caption{Halving Heuristic jumps}\label{alg:hjumpA}
\begin{algorithmic}[1]
\State {$explored \leftarrow 0$}
\State{$temp  \leftarrow \frac{\beta - explored}{2}$}
\If{$temp \le 2\pi - \beta$}
\State{$\alpha_i \leftarrow temp$}
\Else
\State{$ jump \leftarrow 2\pi - \beta$}
\EndIf
\State{$explored \leftarrow explored + jump$}
\end{algorithmic}
\end{algorithm}

Note that the calculations of jumps $\alpha_i$ can be incorporated within Algorithm~\ref{alg:kjumpA} and do not need to be computed in advance. As the maximum number of jumps can be part of the input, Algorithm~\ref{alg:hjumpA} can be performed only for $k$ many landings within the fence (see Algorithm~\ref{alg:kjumpA}), or as long as the the jump step does not drop below a given threshold. Interestingly, the definition of step sizes on the fly by Algorithm~\ref{alg:hjumpA} evenn allows for $k=\infty$. That would correspond to the theoretical case that the robot makes an infinite number of jumps for which all landings happen within the fence. Still, the time for the robot to reach the endpoint of the fence would be finite (Zeno's paradox). 

The process above fully determines the jump step of the $t$-th jump as a function of $\beta$, for every $t=1,\ldots, k$, and for every $k$. In what follows we provide an analytic description of these values so that we can analyze the performance of the algorithm. The lemma below will allow us to derive later a nicer closed formula for the jump steps of the halving Algorithm. 

\begin{lemma}\label{lem: jump steps halving alg}
Let $h_t = \frac{2 \pi (t+1 )}{t+2}$ for $t\geq 1$ and $h_0=0$. For any $\beta \in (0, 2\pi)$, the value of the $i$-th jump in the Halving Algorithm equals
$$
\alpha_i =
\left\{
\begin{array}{ll}
\frac{j\beta - (j-1)2\pi}{2^{i-j+1}}&~\mbox{if}~\beta \leq h_i ~\mbox{and}~\beta \in (h_{j-1},h_j] \\
2\pi-\beta&~\mbox{if}~\beta > h_i \\
\end{array}
\right.
$$
\end{lemma}
\begin{proof} 
We will derive the promised formulas from scratch, without relying on the statement of the lemma. First, note that the process above defines natural threshold values $h_i$ for $\beta$, after which the $i$-th jump step $\alpha_i$ becomes $2\pi-\beta$. In particular, the value of $\alpha_i$ will depend on which interval $(h_{j-1},h_j]$ value $\beta$ belongs to, where $j=1, \ldots k$, and with the understanding that $\alpha_i=2\pi-\beta$ if $\beta > h_i$. Therefore, it is natural to introduce notation 
$$
A_{i,j} := \alpha_i, ~~\mbox{when}~\beta \in (h_{j-1},h_j]
$$

It is easy to see that if $A_{i,j}=2\pi-\beta$ then $A_{r,j}=2\pi-\beta$ for all $r=1, \ldots, i-1$, and in general that $A_{i,j}=2\pi-\beta$ whenever $j\geq i+1$. In other words, $A(i,i)$ is the last expression (in $\beta$) for $\alpha_i$ before it becomes $2\pi-\beta$, while all previous jump steps should be equal to $2\pi-\beta$. Since at every step, the algorithm attempts a jump of half the unexplored fence, 
right before the $i$-th jump there has been explored a total of $(i-1)(2\pi-\beta)$ part of the fence. Hence, 
$$
A_{i,i}=\frac{\beta-(i-1)(2\pi-\beta)}{2} = \frac{i\beta - (i-1)2\pi}2
$$
The threshold $h_i$ is determined by requiring that $A_{i,i}\leq 2\pi-\beta$, from which we obtain that 
$$
h_i=\frac{2 \pi (i+1 )}{i+2}
$$
which is indeed increasing in $i$. 

Our next claim is that 
$$
A_{i,j} = \frac{j\beta - (j-1)2\pi}{2^{i-j+1}},~\mbox{for all}~i\geq j.
$$
The proof is by induction on $i-j$. 
Indeed, the claim is true when $i=j$. So assume that $i=j+t$ for some $t\geq 1$. The explored part of the fence up to the first $(i-1)$ jumps is equal to 
\begin{align*}
\sum_{r=1}^{i-1} A_{r,j} 
=& \sum_{r=1}^{j-1} A_{r,j} +  \sum_{r=j}^{j+t-1} A_{r,j}   \\
=& (j-1)(2\pi-\beta) + \sum_{r=j}^{j+t-1}\frac{j\beta - (j-1)2\pi}{2^{r-j+1}} \\
=& (j-1)(2\pi-\beta) + \left( 1- \frac1{2^t}\right) \left( 2\pi(j-1) - j\beta\right) \\
=&\frac{ \beta \left(2^t-j\right)+2 \pi  (j-1)}{2^t}.
\end{align*}

According to the Halving Algorithm, the $i$-th jump step will be exactly half of the unexplored fence, if that value does not exceed $h_j$. Indeed, the candidate step size is
\begin{align*}
\frac{\beta - \frac{ \beta \left(2^t-j\right)+2 \pi  (j-1)}{2^t}}{2}
&=
\frac{j \beta - (j-1)2\pi}{2^{t+1}}\\
&=
\frac{j \beta - (j-1)2\pi}{2^{i-j+1}} \\
&= A_{i,j}.
\end{align*}

Finally, for this jump to be valid, we need to show that $0<A_{i,j}\leq2\pi-\beta$. To that end, recall that $A_{i,j}$ corresponds to the $i$-th jump when $h_{j-1}< \beta \leq h_j$, i.e. when $\frac{2 \pi j}{j+1} < \beta \leq \frac{2 \pi (j+1 )}{j+2}$. Note that we are in the case where $i\geq j$, and so we have
$
\beta > \frac{2 \pi j}{j+1} > \frac{2 \pi (j-1)}{j}
$ and hence $A_{i,j}>0$. Also,
\begin{align*}
2\pi-\beta - A_{i,j} 
\geq & 2\pi - \frac{2 \pi (j+1 )}{j+2} - \frac{j \beta - (j-1)2\pi}{2^{i-j+1}} \\
\geq &  2\pi - \frac{2 \pi (j+1 )}{j+2} - \frac{j \beta - (j-1)2\pi}{2} \\
= &
\frac{j (2 \pi  (j+1)-\beta (j+2))}{2 (j+2)},
\end{align*}
which is again non negative, since $\beta \leq \frac{2 \pi (j+1 )}{j+2}$, as wanted. 
This proves Lemma~\ref{lem: jump steps halving alg}.
\end{proof}

We are now ready to present a closed formula for the jump steps of the Halving Algorithm along with its performance, as a function on the number of jumps.

\begin{theorem}\label{thm: simpler jump steps halving alg}
For any $\beta \in (0,2\pi)$, let $\rb := \max\left\{  \frac{2\beta-2\pi}{2\pi-\beta} , 1\right\}$. The value of the $i$-th jump ($i\geq 1$) in the Halving Algorithm equals
$$
\alpha_i =
\left\{
\begin{array}{ll}
2\pi-\beta&~\mbox{if}~i < \rb \\
\frac{2\pi-\crb(2\pi-\beta)}{2^{i-\crb+1}}&~\mbox{otherwise}
\end{array}
\right.
$$
\end{theorem}
\begin{proof} 
According to Lemma~\ref{lem: jump steps halving alg}, $\beta>h_i$ is satisfied as long as $i < \frac{2\beta-2\pi}{2\pi-\beta}$. Hence, if $i < \rb$ we have $\alpha_i = 2\pi-\beta$. For the same reason $\beta \leq h_i$ if and only if $i\geq \rb$, so $\rb$ is the smallest integer for which $\beta \leq h_{\rb}$, meaning that $\beta \in ( h_{\rb-1},  h_{\rb}]$. Therefore, again by Lemma~\ref{lem: jump steps halving alg}, we set $j=\crb$ (and rearrange the terms) to derive the promised formula. This proves Theorem~\ref{thm: simpler jump steps halving alg}.
\end{proof}

In Figure~\ref{fig: SomeHalvingJumps}, we depict the behaviour of the decreasing sequence $\alpha_i$ (in $i$) as a function of $\beta$ for a $i=1,2,3,4,5$ and $6$. Notably, for every $k$ there is some threshold value of $\beta$, after which $\alpha_i=2\pi-\beta$ for all $i \leq k$.  
\begin{figure}[h!]
	\centering
	\includegraphics[height=6cm]{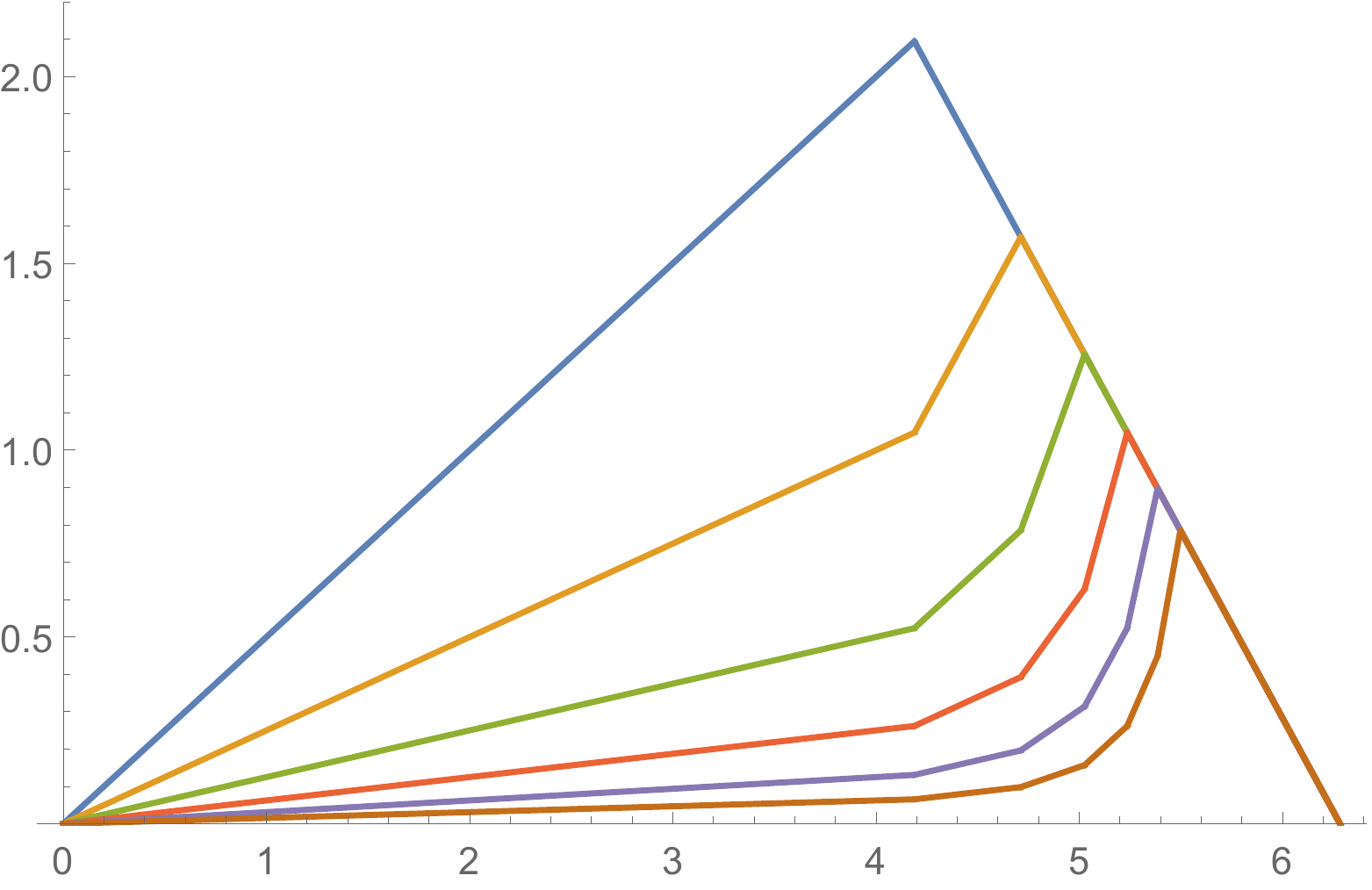}
	\caption{The plot jumps $\alpha_1\geq \alpha_2 \geq  \ldots \geq  \alpha_6$ of the Halving Algorithm as a function of $\beta$. }
	\label{fig: SomeHalvingJumps}
\end{figure}

Finally, we use the closed formula for the jump steps to derive a closed formula for the cost of the Halving \alg{k}. The main idea for the proof is to show that the worst case for the algorithm is when all $k$ jumps fall within the fence, and that the performance is strictly decreasing in $k$. This is what the next lemma establishes. 

\begin{lemma}\label{lem: worst case for halving algo}
The Halving $\alg{k}$ incurs the maximum possible cost when all jump landings (including the basic one) fall within the fence. 
\end{lemma}
\begin{proof} 
For the values of $\alpha_i$ as defined in Theorem~\ref{thm: simpler jump steps halving alg}, we will show that the worst configuration is when all $k$ jump landings (together with the basic one) fall within the fence. In the language of Lemma~\ref{lem: cost of k-jump alg} we will show that $c^k_t < c^k_{t+1}$ for all $t=0, k$. Also note that $c^k_{k+1}$ is decreasing in $k$, in fact no matter what the jump steps are, since $x-\sinn{x/2}>0$, for all $x>0$, which concludes the lemma.

As already claimed, it is immediate that $c^k_0< c_i$, for all $i$, since the cost $c^k_0$ is incurred exactly when the basic landing falls outside the fence. Due to the fact that the robot has knowledge of the length $\beta$, the robot can fully avoid the fence by jumping over it. Next, according to Lemma~\ref{lem: cost of k-jump alg} we have that 
$$c^k_t - c^k_{t-1} = 4 \sinn{\alpha_t/2} - 2\sinn{\alpha_{t-1}/2},$$
for all $t=1, \ldots,k$. In particular, for $t< \rb$ the jump steps remain equal to $2\pi-\beta>0$, and hence $c^k_t - c^k_{t-1}>0$ for all $t< \rb$.

When $t=\crb$ we have
\begin{align*}
c^k_\crb - c^k_{\crb-1}
&= 4 \sinn{\alpha_\crb/2} - 2\sinn{(2\pi-\beta)/2} \\
&~~= 4 \sinn{\frac{\pi}2-\frac{\crb}{4}(2\pi-\beta)} - 2\sinn{\beta/2} \\
&~~\geq 4 \sinn{\frac{\pi}2-\frac{\rb}{4}(2\pi-\beta)} - 2\sinn{\beta/2} \\
&~~\geq 4 \sinn{\frac{\pi}2-\frac{\frac{2\beta-2\pi}{2\pi-\beta}}{4}(2\pi-\beta)} - 2\sinn{\beta/2} \\
&~~= 4 \sinn{\beta/2} - 2\sinn{\beta/2} \geq 0.
\end{align*}
When $t\geq \crb +1$, the jump steps drop by a factor of two in each iteration. Since for all $x>0$ we have that $2\sinn{x/4}-\sinn{x/2}>0$, we obtain easily that $c^k_t - c^k_{t-1}>0$, for $t= \rb+1, \ldots,k$.

Hence it remains to show that $c^k_{k+1}>c_k$. To that end assume that $k\geq \crb$. Then we have 
\begin{align*}
c^k_{k+1} - c^k_{k}
&=\beta-\sum_{i=1}^k \alpha_i - 2\sinn{\alpha_k/2} \notag \\
&=
\beta- \left( (\crb-1)(2\pi-\beta) + 
\sum_{i=0}^{k-\crb} \frac {\alpha_\crb}{2^i} \right) - 
2\sinn{\alpha_k/2} \notag \\
&=
\beta- \left( (\crb-1)(2\pi-\beta) + 
2\alpha_\crb \left(1-\frac1{2^{k-\crb+1}}\right) \right)  - 2\sinn{\alpha_k/2} \notag \\
&=
\beta- \left( (\crb-1)(2\pi-\beta) + 
2\big(\pi-\frac{\crb}2(2\pi-\beta) \right) \left(1-\frac1{2^{k-\crb+1}}\right) \big) - 2\sinn{\alpha_k/2} \notag \\
&=\frac{2\pi - \crb(2\pi-\beta)}{2^{k-\crb+1}}- 2\sinn{\alpha_k/2}
\label{equa: last case}.
\end{align*}

Recall that $\alpha_k=\frac{2\pi-\crb(2\pi-\beta)}{2^{k-\crb+1}}$, so the last expression is non negative since $x-\sinn{x/2}>0$ for all $x>0$. 
Finally, for the case that $k< \crb$, and since $k$ is an integer, we have $k\leq \rb$, and so 
\begin{align*}
c^k_{k+1} - c^k_{k}
&=\beta - k (2\pi-\beta) - 2\sinn{\beta/2}\\
&\geq \beta - \rb (2\pi-\beta) - 2\sinn{\beta/2}\\
&= 2\pi-\beta - 2\sinn{\beta/2}.
\end{align*}
The last expression is non negative since $\sinn{\beta/2}=\sinn{(2\pi-\beta)/2}$ and $x-2\sinn{x/2}>0$, for all $x>0$. 
\end{proof}

We are ready to conclude with the cost of the Halving \alg{k}.

\begin{theorem}\label{cost of halving algo k}
The cost of the Halving \alg{k} is strictly decreasing with $k$ and it equals 
\begin{align*}
&1+2\pi +\frac{\crb(2\pi-\beta)-2\pi}{2^{k-\crb+1}} + 2(\crb-1) \sinn{\beta/2}  + 
2\sum_{i=0}^{k-\crb} \sinn{\frac{\alpha_\crb}{2^{i+1}}},
\end{align*}
where $\rb = \max\left\{  \frac{2\beta-2\pi}{2\pi-\beta} , 1\right\}$ and $\alpha_\crb=\pi- \frac{\crb}2(2\pi-\beta)$.
\end{theorem}
\begin{proof} 
Using the terminology of Lemma~\ref{lem: cost of k-jump alg}, and by Lemma~\ref{lem: worst case for halving algo}, the cost of the Halving \alg{k} equals 
$$
c^k_{k+1} = 1+2\pi - \beta - \sum_{i=1}^k \left(\alpha_i- 2\sinn{\alpha_i /2} \right) 
$$
where the jump steps are as determined in Theorem~\ref{thm: simpler jump steps halving alg}. From the expression above, it is immediate that the cost is strictly increasing in $k$, as long as all jump steps are positive. Next, we compute the summation in parts. We have, 

\begin{align*}
\beta +  \sum_{i=1}^k \alpha_i
&= \beta + \sum_{i=1}^{\crb-1} \alpha_i + \sum_{i=\crb}^{k} \alpha_i \\
&= \beta + (\crb-1) (2\pi-\beta)  + \sum_{i=0}^{k-\crb} \frac{\alpha_\crb}{2^i} \\
&= \beta + (\crb-1) (2\pi-\beta)  + \alpha_\crb \left( 2 - \frac1{2^{k-\crb}} \right) \\
&= \beta + (\crb-1) (2\pi-\beta)  + \left(\pi- \frac{\crb}2(2\pi-\beta)\right)  \left( 2 - \frac1{2^{k-\crb}} \right) \\
 & = \frac{\crb(2\pi-\beta)-2\pi}{2^{k-\crb+1}}.
\end{align*}

Finally, 
\begin{align*}
 \sum_{i=1}^k \sinn{\alpha_i/2}
&=
\sum_{i=1}^{\crb-1} \sinn{\alpha_i/2} + \sum_{i=\crb}^{k} \sinn{\alpha_i/2} =  (\crb-1) \sinn{\beta/2}  + \sum_{i=0}^{k-\crb} \sinn{\frac{\alpha_\crb}{2^{i+1}}}
\end{align*}
Putting the two expression together gives the promised cost. 
This proves Theorem~\ref{cost of halving algo k}.
\end{proof}

Figure~\ref{fig: SomeHalvingPerf} summarizes the cost of the Halving \alg{k} for $k=1,2,3$ and $4$. 

\begin{figure}[h!]
	\centering
	\includegraphics[height=5.5cm]{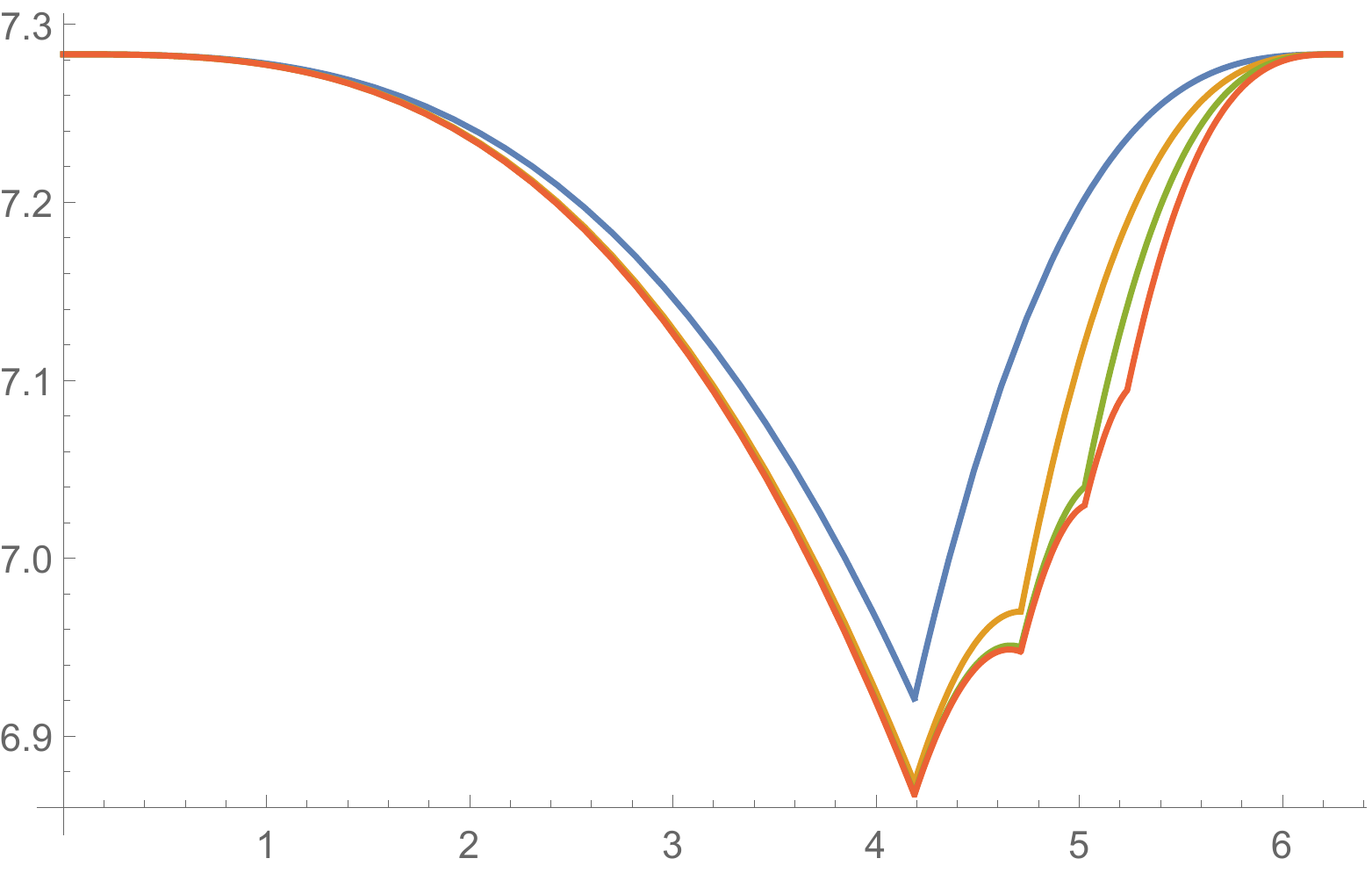}
	\caption{The performance of the Halving Algorithm for 1,2,3 and 4 jumps (decreasing in the number of jumps, respectively) as a function of $\beta$.}
	\label{fig: SomeHalvingPerf}
\end{figure}

\section{Some Optimal \alg{k}s \& Comparison}
\label{sec: comparison}

It is apparent from Lemma~\ref{lem: cost of k-jump alg} that choosing the optimal jump steps $\alpha_1, \ldots, \alpha_k$ amounts to solving the involved optimization problem $\min_{\alpha_1, \ldots, \alpha_k} \max_{t=1,\ldots,k+1}\{c^k_t\}$, where $\alpha_i \leq \min \left\{\pi, 2\pi-\beta \right\}$. In this section we contrast the choices and performance of the Halving \alg{k} with the choices and performance of the Optimal \alg{k} that was obtained using optimization software packages, for $k\leq 3$ (except from $k=1$ whose formal analysis appears in Section~\ref{sec: opt 1-jump}). 
Our findings are summarized in the figures below. 
\begin{figure}[htb!]
	\centering
	\includegraphics[height=6cm]{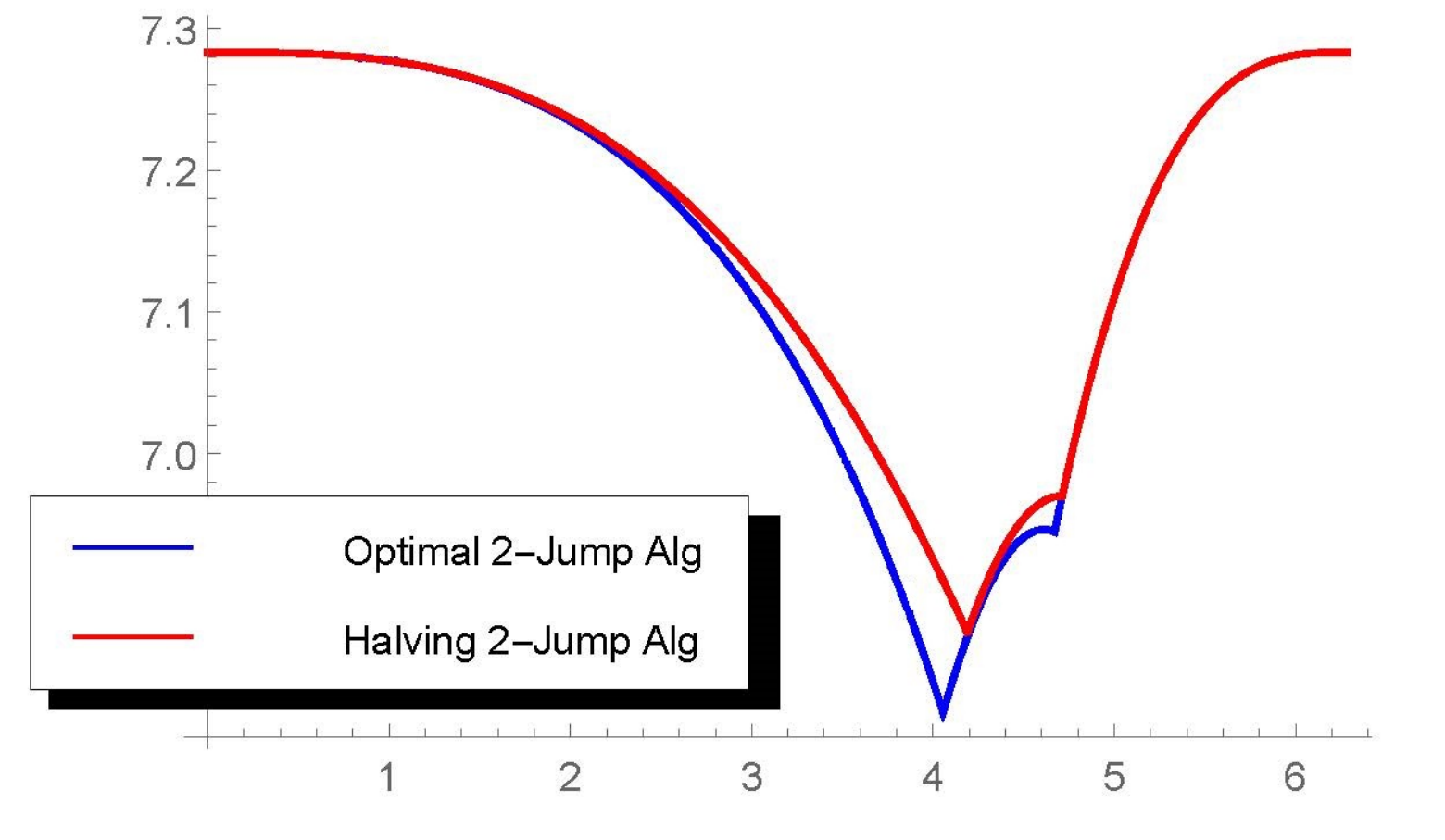}
	\caption{Performance comparison between the Optimal \alg{2} and the Halving \alg{2}, as a function of $\beta$.}
	\label{fig: 2jumpOptHalvPerformance}
\end{figure}

\begin{figure}[htb!]
	\centering
	\includegraphics[height=6cm]{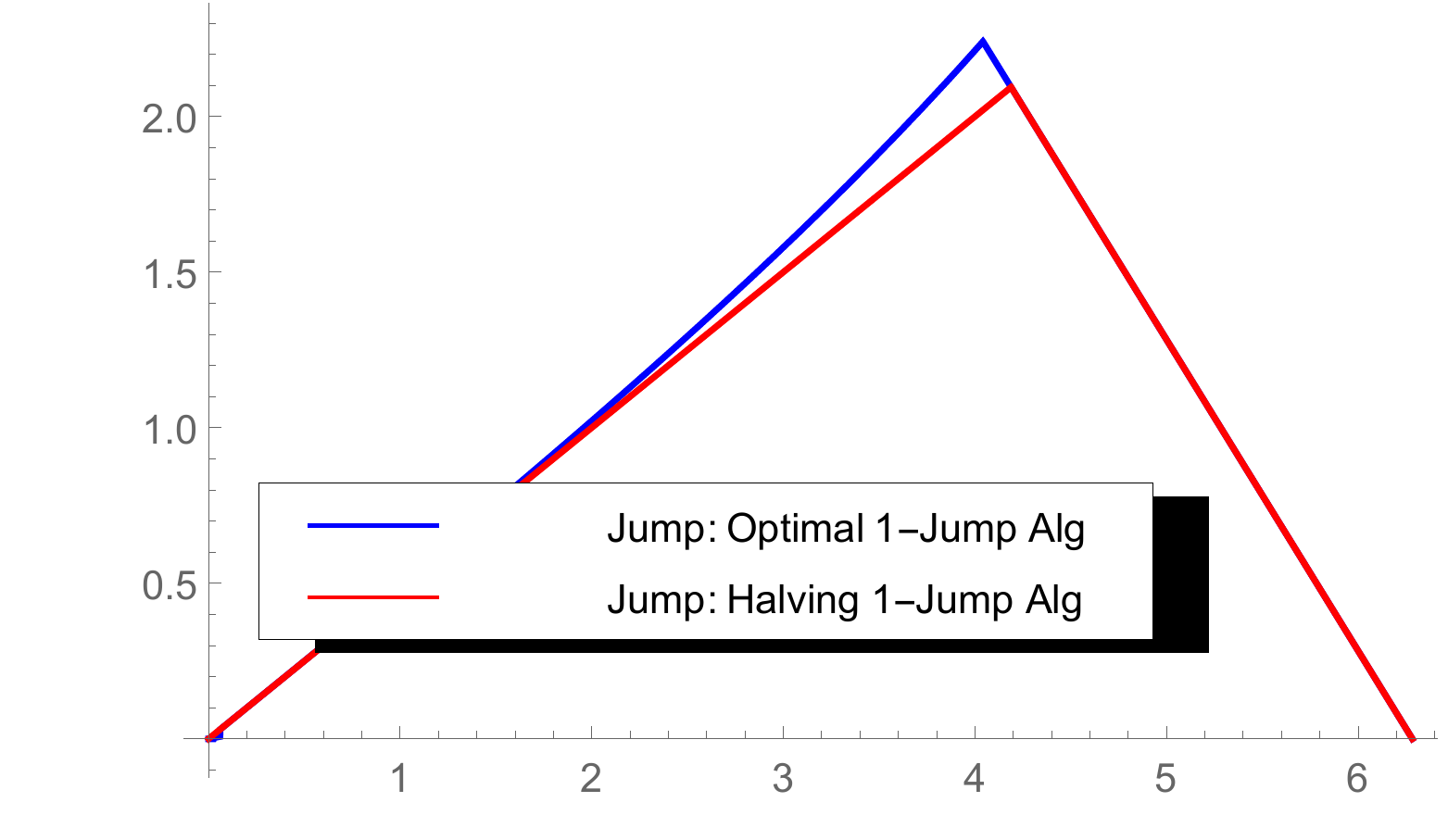}
	\caption{Comparison of jump choices between the Optimal \alg{1} and the Halving \alg{1}, as a function of $\beta$.}
	\label{fig: 1jumpOptHalvJumps}
\end{figure}

\begin{figure}[htb!]
	\centering
	\includegraphics[height=6cm]{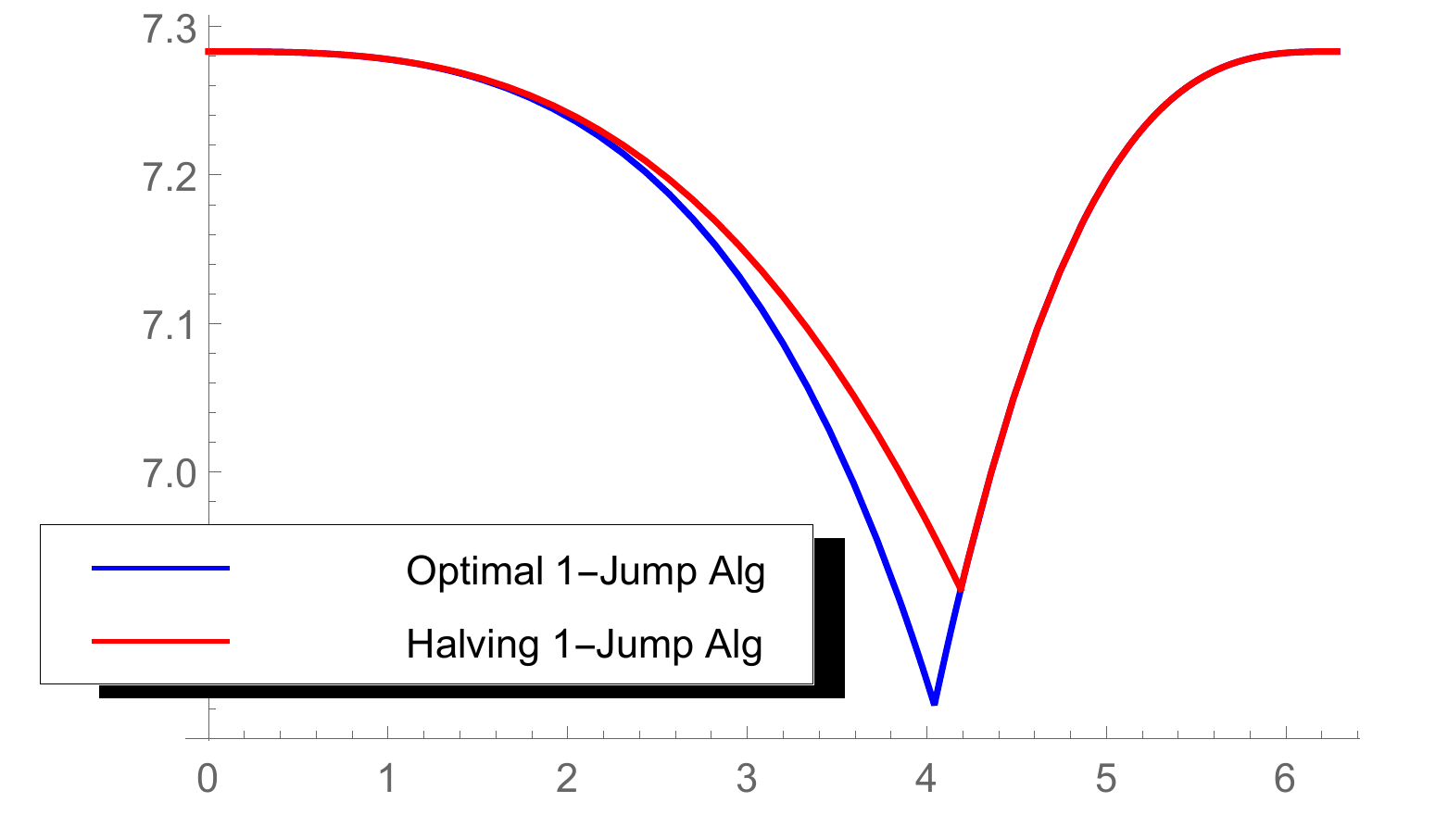}
	\caption{Performance comparison between the Optimal \alg{1} and the Halving \alg{1}, as a function of $\beta$.}
	\label{fig: 1jumpOptHalvPerformance}
\end{figure}

\begin{figure}[htb!]
	\centering
	\includegraphics[height=6cm]{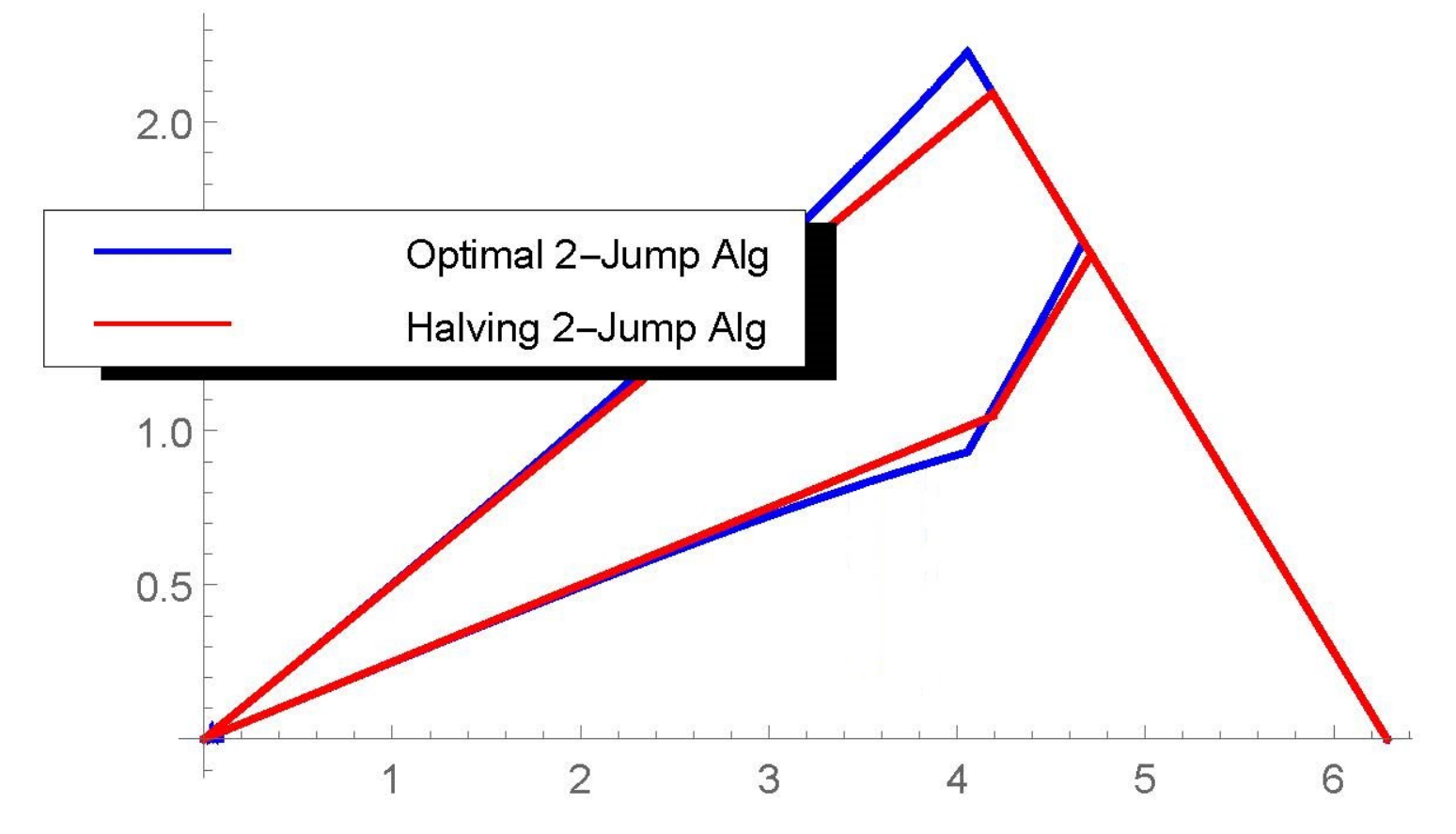}
	\caption{Comparison of the two jump choices between the Optimal \alg{2} and the Halving \alg{2}, as a function of $\beta$. The first jump of each algorithm is always no smaller than the second one, and eventually they all attain the value $2\pi-\beta$.}
	\label{fig: 2jumpOptHalvJumps}
\end{figure}

\begin{figure}[htb!]
	\centering
	\includegraphics[height=6cm]{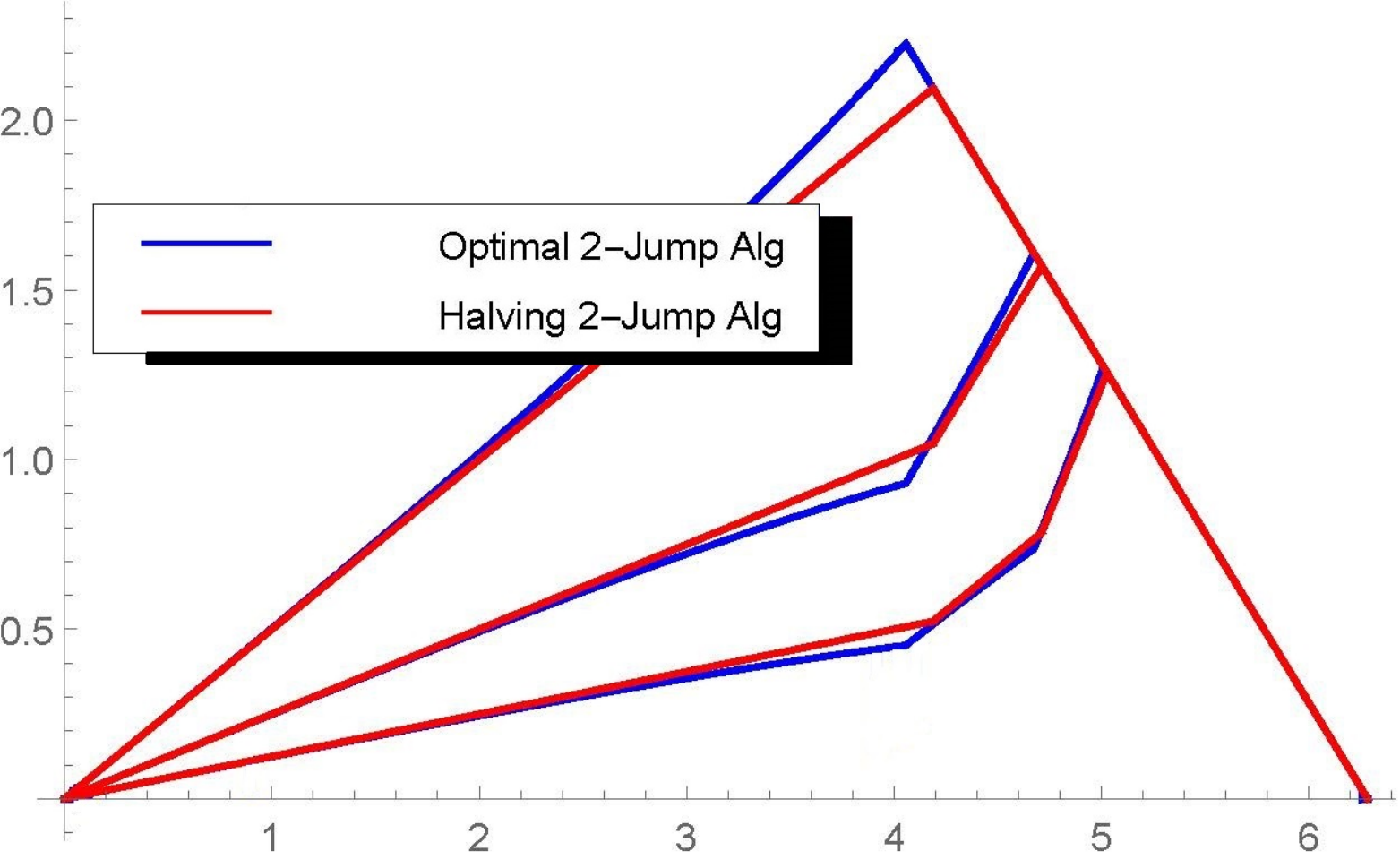}
	\caption{Comparison of the three jump choices between the Optimal \alg{3} and the Halving \alg{3}, as a function of $\beta$. For both algorithms, the first jump of each algorithm is always no smaller than the second one, which is no smaller than the third one. Eventually they all attain the value $2\pi-\beta$.}
	\label{fig: 3jumpOptHalvJumps}
\end{figure}

\begin{figure}[htb!]
	\centering
	\includegraphics[height=6cm]{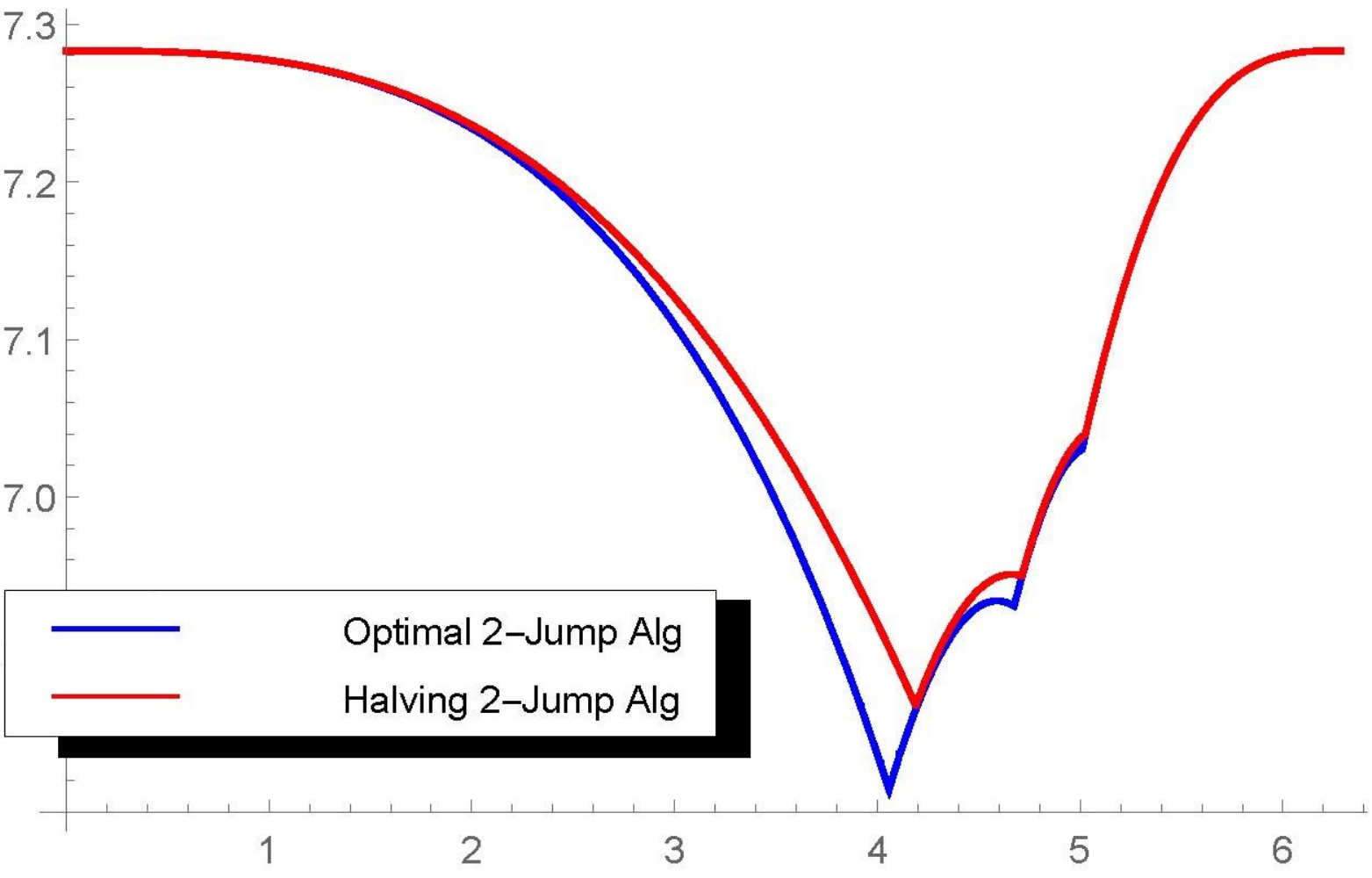}
	\caption{Performance comparison between the Optimal \alg{3} and the Halving \alg{3}, as a function of $\beta$. Notably, performance is nearly the same for all values of $\beta>5$. The bigger discrepancy is observed for values of $\beta$ close to 4, for which also the jump steps between the two algorithms exhibit the larger gaps (see Figure~\ref{fig: 3jumpOptHalvJumps}).}
	\label{fig: 3jumpOptHalvPerformance}
\end{figure}

For $k=1,2,3$ we numerically compute the optimal \alg{k} (note that for $k=1$ the rigorous analysis appears in Section~\ref{sec:1-jump}). Then, we contrast the performance of the optimal and of the Halving algorithm (for the same number of jumps), as well as contrasting their corresponding jump steps. The numerical calculations indicate that the choices of the Halving algorithm are nearly optimal for a wide spectrum of $\beta$ (with the largest discrepancy for $\beta \approx \gamma$). Interestingly, for larger values of $\beta$, the choices of the Halving algorithm are nearly optimal that also reflects on the cost of the two algorithms which becomes nearly identical. More importantly, experiments indicate that for large values of $\beta$, the optimal choices for $k$ jump steps is to make all equal to $2\pi-\beta$, which is also the choice of the Halving algorithm. 

\section{Conclusion}
\label{sec:conclusion}

In this paper we investigated a new search problem for a mobile robot to find a stationary target placed at an unknown location at distance $1$, in the presence of a fence placed on the perimeter of a unit disc. First we determined the optimal $1$-Jump algorithm for the robot to find the target. Then we provided a generic description of $k$-Jump algorithms and analyzed their cost depending on the jump landings. Subsequently we analyzed the Halving $k$-Jump algorithms, where $k$ is the max number of jumps the robot makes so as to overcome the fence and find the target. Several interesting questions remain open, when e.g., there are multiple fences on the perimeter of the disc, and the robot's speed changes when traversing a fence.
 
 
\small
\bibliographystyle{abbrv}
\bibliography{refs}

%
%


\end{document}